\documentclass{article}
\usepackage{url}
\usepackage{graphics}
\usepackage{tikz,pgffor,amsmath,amssymb,amsthm}
\usetikzlibrary{calc}

\DeclareMathOperator{\KS}{\mathit{C}\,}
\DeclareMathOperator{\KP}{\mathit{K}\,}
\newcommand{\0}{\mathbf{0'}}
\DeclareMathOperator{\KPH}{\mathit{K}^\0\,}

\newcommand{\rar}{\rightarrow}

\newtheorem{theorem}{Theorem}

\newtheorem{lemma}[theorem]{Lemma}

\newtheorem{question}[theorem]{Question}

\begin{document}

\title{Prefix and plain Kolmogorov complexity characterizations of $2$-randomness:\\ simple proofs}
\author{Bruno Bauwens\footnote{
  LORIA, Universit\'e de Lorraine, 
  615-B248 Rue du Jardin Botanique,
  54600 Vand\OE vre-l\`es-Nancy, France, www.bcomp.be.
  The research (carried out in spring 2012) was supported by the Portuguese science foundation FCT
  SFRH/BPD/75129/2010 and NAFIT ANR-08-EMER-008-01 projects.
  The author is grateful to Laurent Bienvenu and Nikolay Vereshchagin for there
  simplified proofs of Conidis' result~\cite{personalKoliaLaurent12} (see
  also~\cite{limitComplexitiesOnceMore}).
  The author is also grateful to Sasha Shen for useful discussion and the full rewrite of the text
  presented here!
  }
}

\date{}

\maketitle  

\begin{abstract}
  Joseph Miller~\cite{Miller2randC} and independently Andre Nies, Frank Stephan and Sebastiaan
  Terwijn~\cite{Nies2rand} gave a complexity characterization of $2$-random sequences in terms of
  plain Kolmogorov complexity $C(\cdot)$: they are sequences that have infinitely many initial
  segments with $O(1)$-maximal plain complexity (among the strings of the same length). 

  Later Miller \cite{Miller2randK} showed that prefix complexity $K(\cdot)$ can also be used in a
  similar way: a sequence is $2$-random if and only if it has infinitely many initial segments with
  $O(1)$-maximal \emph{prefix} complexity (which is $n+\KP(n)$ for strings of length~$n$). 

  The known proofs of these results are quite involved; in this paper we provide simple direct proofs for both of them.

  In~\cite{Miller2randC} Miller also gave a quantitative version of the first result: the
  $\mathbf{0}'$-randomness deficiency of a sequence $\omega$ equals $\liminf_n [n -
  \KS(\omega_1\dots\omega_n)] + O(1)$.  (Our simplified proof can also be used to prove this.)
  We show (and this seems to be a new result) that a similar quantitative result
  is also true for prefix complexity: $\mathbf{0}'$-randomness deficiency equals 
  $\liminf_n [n + \KP(n) - \KP(\omega_1\dots\omega_n)]+ O(1)$. 
\end{abstract}

\section*{Introduction}

The connection between complexity and randomness is one of the basic ideas that motivated the
development of algorithmic information theory and algorithmic randomness theory. However, at first
the definition of complexity (plain complexity of a bit string, introduced by
Ray Solomonoff~\cite{solomonoffI} and
Andrei Kolmogorov~\cite{Kolmogorov65} as the minimal length of a program that produces this string) and the definition
of randomness (given by Per Martin-L\"of~\cite{Martinlof66}) were given separately, and only later some connections
between them became clear.

Leonid Levin~\cite{levinPpi,complexityOfComplexity} and 
later Gregory Chaitin~\cite{chaitin75} introduced a modified version of complexity, called \emph{prefix}
complexity and denoted usually by $\KP(\cdot)$, that corresponds to self-de\-limit\-ing programs. It
turned out (see the papers of Claus-Peter Schnorr~\cite{SchnorrProcess},
Levin~\cite{levin73}, Chaitin~\cite{chaitin75}) that a bit
sequence $\omega=\omega_1\omega_2\ldots$ is Martin-L\"of random if and only if $\sup_n
[n-\KP(\omega_1\ldots\omega_n)]$ is finite. Moreover, this supremum coincides with randomness
deficiency (a quantitative version of Martin-L\"of definition of randomness suggested by Levin and
Peter Gacs, see~\cite{GacsExactExpressions}).

Let us recall the definition of randomness deficiency since it is less known compared to other
notions of algorithmic information theory. By $\Omega$ we denote the Cantor space of infinite bit
sequences. 

\begin{itemize}

\item A \emph{basic function} is a function $f\colon\Omega\to\mathbb{Q}^+$ whose value $f(\omega)$
  is a non-negative rational number that depends on a \emph{finite} initial prefix of $\omega$ of
  some length. Basic functions are constructive objects, so we can speak about computable sequences
  of basic functions.

\item A \emph{lower semicomputable function} is a function $f\colon\Omega\to\overline{\mathbb{R}}^+$
  (values are non-negative reals and $+\infty$) that is a pointwise upper bound of a computable
  sequence of basic functions. Equivalent definition: a sum $\sum h_i(\cdot)$ where $h_i(\cdot)$ is
  a computable sequence of basic functions.

\item A \emph{randomness test} is a lower semicomputable function $t$ such that the integral $\int
  t(\omega) dP(\omega)$ does not exceed~$1$. (Here $P$ is the uniform Bernoulli measure on Cantor
  space that corresponds to independent fair coin tossings.)

\item There exists a \emph{universal} randomness test $u(\omega)$ that exceeds every other one (up
  to $O(1)$-factor). We fix some universal randomness test $\mathbf{u}$. Its logarithm $\log \mathbf{u}(\omega)$ is
  called the \emph{randomness deficiency} of $\omega$ and denoted by $\mathbf{d}(\omega)$. 
  The randomness deficiency is
  defined up to $O(1)$-additive term since different universal tests differ at most by a bounded
  factor.

\end{itemize}

The quantitative version of Schnorr--Levin theorem says that 
$$
    \mathbf{d}(\omega)=\sup_n [n - \KP(\omega_1\ldots\omega_n)]+O(1).
$$
So we can give an equivalent definition of randomness deficiency just as the supremum in the right-hand side of this equation.

This statement looks a bit counterintuitive. One can expect that a sequence is random if its initial
segments (prefixes) have maximal possible complexity (among all strings of the same length). But the
maximal prefix complexity for $n$-bit strings is $n+\KP(n)$, not $n$, up to $O(1)$ additive term. So why we compare
$\KP(\omega_1\ldots\omega_n)$ to $n$, not to $n+\KP(n)$? Or why we consider prefix complexity and
not the plain one, for which the maximal complexity of $n$-bit string is indeed~$n$?

The obstacle here is an old Martin-L\"of observation: for \emph{every} sequence $\omega$ the
difference $n-\KS(\omega_1\ldots\omega_n)$, as well as the difference
$n+\KP(n)-\KP(\omega_1\ldots\omega_n)$,  is unbounded. There are some workarounds, still: for
example, instead of requiring that $n-\KS(\omega_1\ldots\omega_n)$ is bounded for \emph{all} $n$, we
can require it to be bounded for \emph{infinitely many} $n$, i.e., consider sequences such that
$\liminf_n[n-\KS(\omega_1\ldots\omega_n)]$ is finite.\footnote{The other (may be, more natural)
approach is to consider the so-called \emph{monotone} complexity, or \emph{a priori} complexity,
that do not have this problem. We do not consider these complexities in our paper.} It is easy to
see that indeed this $\liminf$ is finite for almost all sequences (except for a set of zero
measure). What are these sequences? 

The answer was found by Joseph Miller~\cite{Miller2randC} and independently by Andre Nies, Frank
Stephan and Sebastian Terwijn~\cite{Nies2rand}. They proved that this class of sequences  coincides
with the class of \emph{$2$-random sequences}, i.e., the sequences that are Martin-L\"of random even
with an oracle for $\mathbf{0}'$ (the halting problem).  The proof in~\cite{Miller2randC} is quite involved, and
the proof in~\cite{Nies2rand} uses special tools from recursion theory (the low basis theorem). Some
other approach was suggested in~\cite{limitComplexitiesRevisited}, 
and later Chris Conidis~\cite{conidis} showed that
one can avoid low basis theorem in this way. Still Conidis' argument is a bit complicated. In
Section~\ref{sec:plain} we provide a simple proof of Conidis' result thus giving a simple proof of
Miller--Nies--Stephan-Terwijn characterization of $2$-random sequences. Extending this argument and
using an effective version of Fatou lemma, we get also a new simple proof for a quantitative version
of this characterization from~\cite{Miller2randC}: 
$$
  \liminf [n-\KS(\omega_1\ldots\omega_n)]=\mathbf{d}^{\mathbf{0}'}(\omega)+O(1).
$$
In the right-hand side $\mathbf{d}^{\mathbf{0}'}$ stands for the randomness deficiency relativized to
$\mathbf{0}'$; this deficiency is finite when $\omega$ is $2$-random.
  
Later Miller \cite{Miller2randK} got a similar result for prefix complexity: a sequence $\omega$ is
$2$-random if and only if $\omega$ has infinitely many initial segments with $O(1)$-maximal
\emph{prefix} complexity (which is $n+\KP(n)$ for strings of length~$n$), i.e., if 
$$
  \liminf [n+\KP(n)-\KP(\omega_1\ldots\omega_n)]
$$
is finite. The original proof was even more complicated than the proof for plain complexity; it used
van Lambalgen theorem about random pairs, Ku\v{c}era -- Slaman result about random lower semicomputable
reals and some other tools. Some simplifications were found by Laurent Bienvenu and others (see Downey and
Hirschfeldt~\cite{Downey}), but even with these simplifications the proof remains quite difficult. In
Section~\ref{sec:prefix} we present a much simpler proof.

Finally, in Section~\ref{sec:prefixStrong} we show that this result also has a quantitative version,  
thus completing the picture:
\begin{eqnarray*}
 \mathbf{d}^{\0}(\omega) & = & \sup n - \KP^\0(\omega_1\dots\omega_n) + O(1) \\
 & = & \liminf [n - \KS(\omega_1\dots\omega_n)] + O(1) \\
 & = & \liminf [n + \KP(n) - \KP(\omega_1\dots\omega_n)] + O(1) \,. 
\end{eqnarray*}
It is not clear whether this quantitative version can be extracted from
Miller's argument. One can raise the question whether the same initial segments have maximal plain
or prefix complexity. In an upcomming paper we show this is not the case: for every $3$-random
sequence, there exist a $c$ and infinitely many prefixes $x$ such that $n- \KS(x) \le c$ and $n +
\KP(n) - \KP(x) \ge \log\log n - c$.

Section~\ref{sec:plain} and \ref{sec:prefix}--\ref{sec:prefixStrong} are (mostly) independent, so the readers interested only in plain or prefix complexity
can proceed directly to the corresponding part of the paper.

\section{Plain complexity and 2-randomness}
\label{sec:plain}

This section is devoted to the Miller--Nies--Stephan--Terwijn characterization of $2$-random
sequences in terms of plain complexity, and it's quantified form: 

\begin{theorem}[Miller]\label{th:2randC}
  \[
    \mathbf{d}^\0(\omega) = \liminf [n - \KS(\omega_1\ldots \omega_n)] + O(1) \,.
  \]
\end{theorem}

First let us reproduce the proof of the easy direction ($\le$). We assume that $\mathbf{d}^\0(\omega)$ equals $d$, and show that $n-\KS(\omega_1\ldots\omega_n)\ge d-O(1)$ for sufficiently large $n$. Since 
    $$\mathbf{d}^\0(\omega)=\limsup n-\KP^\0(\omega_1,\ldots,\omega_n)$$
(we omit $O(1)$ terms here and later) we may assume that 
    $$
\KP^\0(\omega_1\ldots\omega_m)\le m-d
    $$
for some $m$. Then we can use the additivity property\footnote{The direction ($\le$) that we need is quite simple: $\KS(a,b)=\KS(a,b|\KS(a,b))$, and $\KS(u,v|w)\le \KP(u|w)+\KS(v|w)$ by concatenation of the programs.} for plain complexity~\cite{BauwensAdditivity},
  $$\KS(a,b) = \KP(a|\KS(a,b)) + \KS(b|a,\KS(a,b)),$$   
for $a=\omega_1\ldots\omega_m$ and $b=\omega_{m+1}\ldots\omega_n$. Then we have
   $$
\KS(\omega_1\ldots\omega_n)\le \KS(a,b)\le \KP(a|\KS(a,b))+\KS(b|\KS(a,b)).
  $$
The second term does not exceed $|b|$, i.e., $n-m$; it is enough to show, therefore, that the first term is bounded by $m-d$, i.e., by $\KP^\0(\omega_1,\ldots,\omega_m)$. Indeed, the condition $C(a,b)$ tends to infinity as $n\to\infty$, and $\lim_N \KP(x|N)\le \KP^\0(x)$. (Indeed, we can approximate $\0$ making $N$ steps of enumeration, and for large $N$ this is enough.)

Now we switch to the other direction ($\ge$). The qualitative version says that a sequence $\omega$ such that $n-\KS(\omega_1\ldots\omega_n)\to\infty$, is not $\0$-random, and we start by proving this version. So let us assume that $\KS(\omega_1\ldots\omega_n)<n-c$ for all sufficiently large $n$.  To show that $\omega$ is not Martin-L\"of $\0$-random, we need to cover $\omega$ by a $\0$-effectively open set of small measure (uniformly).

Consider the set $U_n$ of sequences $\alpha$ such that $\KS(\alpha_1\ldots\alpha_n)<n-c$. This is an
effectively open set (uniformly in $n$) that has measure at most $2^{-c}$ (since there are less than
$2^{n-c}$ strings of complexity less than $n-c$). We know that our sequence $\omega$ belongs to all
$U_n$ for sufficiently large $n$ (but we do not know the threshold for ``sufficiently large''). It
remains to apply the following result of Conidis~\cite{conidis} (for its applications and discussion
see also~\cite{limitComplexitiesRevisited} where this statement was mentioned as a conjecture, and
the revised version~\cite{limitComplexitiesOnceMore}).

\begin{theorem}[Conidis]\label{th:conidis} 
 Let $\varepsilon>0$ be a rational number and let $U_0,U_1,\ldots$ be a sequence of
 uniformly effectively open sets of measure at most $\varepsilon$ each.  Then for every rational
 $\varepsilon'>\varepsilon$ there exists a $\mathbf{0}'$-effectively open set $V$ of measure at
 most $\varepsilon'$ that contains $\liminf_{n\to\infty} U_n= \bigcup_{N} \bigcap_{n\ge N} U_n$,
 and the $\mathbf{0}'$-enumeration algorithm for $V$ can be effectively found given $\varepsilon$,
 $\varepsilon'$, and the enumeration algorithm for $U_i$.  
\end{theorem}

\begin{proof}
 Let us denote by $U_{k..l}$ the intersection $U_{k}\cap U_{k+1}\cap\ldots\cap U_l$. The set $V$ will be
 constructed as $U_{1..k_1} \cup U_{k_1+1..k_2} \cup\ldots$ for some $\0$-computable sequence
 $k_1<k_2<\ldots$; this guarantees that $V$ is $\0$-effectively open and that $\liminf U_i\subset
 V$. It remains to explain how we choose $k_i$ such that $V$ has measure at most $\varepsilon'$.
 
Let us fix an increasing computable sequence $\varepsilon<\varepsilon_1<\varepsilon_2<\ldots<\varepsilon'$. 
 There exists some $k_1$ such that for every $i>k_1$ the set 
 $$ 
   U_{1..k_1} \cup U_i
 $$ 
 has measure at most $\varepsilon_1$. Indeed,
 if for some $i$ the measure is greater than $\varepsilon_1$, then, adding $U_i$ as a new term in the
 intersection (by increasing $k_1$ up to $i$), we decrease the measure of the intersection at least
 by $\varepsilon_1-\varepsilon$. (If $A\cup B$ has measure greater than $\varepsilon_1>\varepsilon$
 while $B$ itself thas measure at most $\varepsilon$, then $A\setminus B$ has measure at least
 $\varepsilon_1-\varepsilon$, so the measure of $A$ decreases at least by
 $\varepsilon_1-\varepsilon$ after intersecting it with $B$.) If the newly found $k_i$ 
 does not satisfy the condition, we repeat the process. Each time this happens, the measure of the intersection 
 decreases by at least $\varepsilon_1-\varepsilon$, hence this can happen only finitely many times. 
 
 For similar reasons we can then find $k_2$
 such that for every $i$ the set 
 $$ 
  U_{1..k_1} \cup U_{k_1+1..k_2} \cup U_i  
 $$ 
 has measure at most $\varepsilon_2$ for every $i>k_2$. 
 Indeed, the size of $U_{1..k_1} \cup U_i$ is bounded by $\varepsilon_1$, 
 hence if the measure of the set above exceeds $\varepsilon_2$, then there is at least a 
 $(\varepsilon_2 - \varepsilon_1)$-part of $U_{k_1+1..k_2}$ outside $U_{1..k_1} \cup U_i$ (in particular, outside $U_i$). Thus adding $U_i$ as a new term in the
 intersection $U_{k_1+1..k_2}$ decreases its measure by at least $\varepsilon_2-\varepsilon_1$;
 such a decrease may happen only finitely many times. 
 
 We continue this construction for $k_3,k_4$ etc. Note that this construction is $\mathbf{0}'$-computable and the union 
 $$ 
  V=U_{1..k_1} \cup U_{k_1+1..k_2} \cup U_{k_2+1..k_3} \cup\ldots 
 $$ 
 is an $\mathbf{0}'$-effectively open cover of $\liminf U_n$ of measure at most $\varepsilon'$.
\end{proof}

A more careful analysis of this argument allows us to get the statement of Theorem~\ref{th:2randC} in weak form, with logarithmic precision. So we need to modify the argument. First, we formulate a version of Conidis' theorem with functions instead of sets (that also can be considered as a constructive version of Fatou's lemma).

\begin{theorem}\label{th:Fatou-constr}
Let $f_1,f_2,\dots$ be a series of uniformly lower semicomputable functions on Cantor space such
that $\int f_i(\omega)\,d\mu(x)$ does not exceed some rational~$\varepsilon>0$ for all~$i$. Then for every
$\varepsilon'>\varepsilon$ one can uniformly construct a lower $\mathbf{0}'$-semicomputable
function $\varphi$ such that 
$$
  \liminf\,f_n (\omega)\le \varphi(\omega)\text{ \ for every $\omega$, \ and }\int\varphi(\omega)
  d\mu(\omega)\le\varepsilon'.
$$
\end{theorem}

We get the original Conidis' result when $f_i$ are indicator functions of open sets. In fact, the proof remains almost the same. For each function $f_i$ we consider the set $U_i$ below its graph, i.e., the set of pairs $(\omega, u)$ in $\Omega\times\mathbb{R}$ such that $0\le u\le f_i(\omega)$. The measure of this set equals $\int f_i(\omega)\,d\omega$. The intersection/union operations with these sets correspond to min/max operations with the functions. So the same construction as before gives the function
    $$\varphi(\omega) = \sup(f_{1..k_1}(\omega),f_{k_1+1..k_2}(\omega),\ldots)$$
where 
    $$f_{k..l}(\omega)=\min(f_k(\omega), f_{k+1}(\omega),\ldots, f_l(\omega)).$$
It is easy to see that $\liminf_n f_n(\omega)\le\varphi(\omega)$ (note that $\liminf$ operation on functions corresponds to the same operation on sets). Also functions $f_{i..j}$ are lower semicomputable (minimum of a finite family of lowersemicomputable functions is lower semicomputable), and the function $\varphi$ is semicomputable with an oracle that computes the sequence $k_i$. 

Theorem~\ref{th:Fatou-constr} is proved.

Now we use this theorem to show that if $C(\omega_1\ldots\omega_n)<n-c$ for large $n$, 
then $\mathbf{d}^\0(\omega)\ge c-O(1)$. For that we need to construct a $\0$-lower 
semicomputable randomness test that exceeds $2^c$ on all those~$\omega$. 

One may try to let $f_n(\omega)$ be equal to $2^{n-\KS(\omega_1\ldots\omega_n)}$. Then for all
$\omega$ in question we have $f_n(\omega)>2^c$ for large $n$, and $\liminf f_n(\omega)\ge 2^c$. 
If the integrals $\int f_n(\omega)\,d\omega$ were bounded, we could finish the
proof by applying Theorem~\ref{th:Fatou-constr}. However, it is not the case: we know that
$f_n(\omega)$ exceeds $2^k$ on a set of measure at most $2^{-k}$ (for every~$k$), but this is not
enough for the integral bound.

To fix the problem, we change the definition of $f_n$. For a binary string $u$, let us define the function $\chi_{x\Omega}$ that equals $1$ on the extensions of $x$ and equals $0$ otherwise. Its integral is $2^{-|x|}$. Multiplying this function by $2^{|x|-m}$ for some $m$, we get a function with integral $2^{-m}$. Then consider the sum
     $$
f_m(\omega)=\sum_{\{x\mid\KS(x)<m\}} 2^{|x|-m}\chi_{x\Omega}.
     $$
This sum contains less than $2^m$ terms; each has integral $2^{-m}$, so the integral of the sum is
bounded by~$1$. On the other hand, if $\KS(\omega_1\ldots\omega_n)<n-c$ for all large enough $c$,
the sum for $f_m(\omega)$ includes a term of size at least $2^c$ for all sufficiently large $m$. 

This observation finished the proof of Theorem~\ref{th:2randC}.

\section{Prefix complexity and $2$-randomness}
\label{sec:prefix}

In this section we provide a simple proof of the following result of Miller:
\begin{theorem}[Miller]\label{th:miller-pref}
A sequence $\omega$ is 2-random \textup(Martin-L\"of random with oracle $\0$\textup) if and only if 
     $\liminf_n [n + \KP(n) - \KP(\omega_1\ldots \omega_n)]$
is finite.
\end{theorem}

In the next section we will prove a quantitative version of this result: this $\liminf$ equals $\mathbf{d}^\0(\omega)$, and this will require a more complicated proof. However, in one of the directions the quantitative result is equally simple, so we start with this direction.

Let us prove that $\mathbf{d}^\0(\omega)\le\liminf [n+\KP(n)-\KP(\omega_1\ldots\omega_n)]$. We use almost the
same argument as for Theorem~\ref{th:2randC}. Since $\mathbf{d}^\0(\omega)$ is equal to $\liminf_m [m-\KP^\0(\omega_1\ldots\omega_m)$ up to $O(1)$ additive term, we assume that $\KP^\0(\omega_1\ldots\omega_m)=m-d$ and show that $\KP(\omega_1\ldots\omega_n)\le n+\KP(n)-d +O(1)$ for large~$n$.

Let $a=\omega_1\ldots\omega_m$ and $b=\omega_{m+1}\ldots\omega_n$. Using the bound for the prefix complexity of a pair $\KP(u,v)\le\KP(u)+\KP(v|u)+O(1)$ (also in the conditional version), we note that (up to $O(1)$-terms)
\begin{align*}
    \KP(\omega_1\ldots\omega_n)&\le\KP(n)+\KP(\omega_1\ldots\omega_n|n)\le\\
       &\le \KP(n)+\KP(a,b|n)\le\\
       &\le \KP(n)+\KP(a|n)+\KP(b|a,n).
\end{align*}
 It remains to note that 
 \begin{itemize}
 \item the last term does not exceed $m-n$ (the condition is enough to reconstruct $m-n$, and the prefix complexity of a string when its length is given, is bounded by this length);
 \item for sufficiently large $n$ the value of $\KP(a|n)$ does not exceed $\KP^\0(a)$ (the required part of $\0$ can be reconstructed during $n$ enumeration steps).
 \end{itemize}
So, for large $n$ the right-hand side is bounded by 
$$\KP(n)+\KP^\0(a)+n-m\le \KP(n)+(m-d)+n-m=n+\KP(n)-d,$$
as required.

It remains to prove the (qualitative) statement in the other direction:
\begin{quote}\itshape
  Let $\omega$ be a binary sequence such that
      $ 
  \KP(\omega_1\dots\omega_n) - (n + \KP(n)) \to - \infty.
      $
  Then $\omega$ is not $2$-random.
\end{quote}
It will be done in the rest of the section, in several steps.

\subsection{Slow convergence}

Let us start with the following simple definition.
Let $a_i$ and $b_i$ be two series with
non-negative terms. We say that \emph{$a_i$-tails are bounded by $b_i$-tails} if
\[
(a_N + a_{N+1} + \ldots) \leq c (b_N + b_{N+1} + \ldots)
\]
for some $c$ and all $N$. We assume here that $\sum a_i$ converges (but $\sum b_i$ may diverge). 
Reformulation: $a_i$-tails are \emph{not} bounded by $b_i$-tails if the ratio
\[
\frac{a_N + a_{N+1} + \ldots}{b_N + b_{N+1} + \ldots} 
\]
is unbounded. 

\textbf{Examples}:

    1. Let $\mathbf{m}(i)$ be the (discrete) a priori probability of $i$, the maximal (up to a
    constant) lower semicomputable converging series; we may let $\mathbf{m}(i)=2^{-\KP(i)}$ (see
    e.g., \cite{LiVitanyi} or \cite{ShenIntro}). Then the tails of every convergent computable
    series $\sum a_i$ are bounded by the tails of the series $\sum\mathbf{m}(i)$. Indeed,
    $a_i \le O(\mathbf{m}(i))$ implies the same relation for tails.

    2. On the other hand, for every lower semicomputable series there
    exist a computable series with rational terms that
    has the same limit and has bigger tails (that bound the tails of the first one). Indeed,
    each lower semicomputable term can be split into
    a sum of a computable series, and we can add all
    the summands (for all terms) one by one; this 
    delay can only increase the tails. Therefore, being bounded by tails of some convergent computable series is equivalent to being bounded by the tails of $\sum \mathbf{m}(i)$.

\subsection{Lower semicomputable tests and 2-randomness}

Remind from the introduction that Martin-L\"of randomness can be defined using randomness tests (lower semicomputable non-negative
functions on the Cantor space that have integral at most $1$, see the Introduction).  It turns out
that lower semicomputable tests can
be used in a more ingenious way to show that
some sequence is \emph{not 2-random} (not ML-random
relative to the halting problem).

Let $f_i(\cdot)$ be a sequence of (uniformly) lower
semicomputable non-negative functions on $\Omega$. Assume that the sum
$\sum_i\int f_i$ is finite. 
Thus $\sum_i f_i(\cdot)$ is a lower semicomputable test, 
and every sequence $\omega$ such that $\sum_i f_i(\omega)$ diverges, 
is not ML-random. Moreover,
the following statement (where both the condition
and the claim are weaker) is true:

\begin{lemma}\label{lem:coverSlowConvergence}
  If  the tails of the series $\sum_i f_i(\omega)$ are not bounded by any computable series, 
  then $\omega$ is not $\mathbf{0}'$-random.
\end{lemma}

As we have seen, we may use for comparison
the series $\sum_i \mathbf{m}(i)$ instead of computable series.

\begin{proof}
Without loss of generality we may assume
that $f_1(\cdot)$, $f_2 (\cdot)$, $\ldots$ 
is a computable sequence of basic functions 
(splitting each semicomputable term
into a sum of computable terms, we only increase
the tails).

To show that every $\omega $ with this property 
(very slow convergence) is not $\mathbf{0}'$-random, we need to construct for
every rational $\varepsilon > 0$ a $\mathbf{0}'$-effectively
open set of measure at most $\varepsilon $ that covers (all
such) $\omega$. This construction goes as follows.
Consider computable increasing sequences of basic functions $S_i\colon \Omega \to \mathbb{Q}$ and rational numbers $t_i$ (``thresholds'') constructed in the following way.
We start with zero function $S_0$ and zero threshold
$t_0$. Then for each $i = 1, 2, 3, \dots$ we do the following
steps:

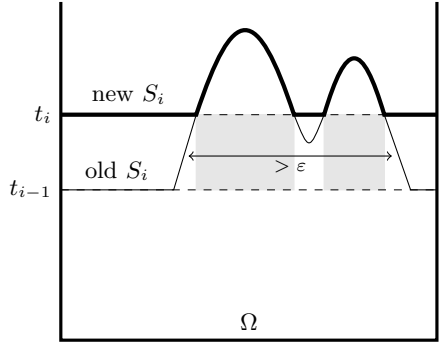
\begin{figure}[h]
  \begin{center}
    \begin{tikzpicture}
      \draw[very thick] (0,4.5) -- (0,0) -- node [anchor=south] {\small{$\Omega$}} (5,0) -- (5,4.5);
      \coordinate (A) at (1.8,3);
      \coordinate (B) at (3.1,3);
      \coordinate (C) at (3.5,3);
      \coordinate (D) at (4.3,3);
      \filldraw[gray!20] (A) rectangle (3.1,2);
      \filldraw[gray!20] (C) rectangle (4.3,2);
      \draw[ultra thick] (0,3) -- node[above] {\small{new $S_i$}} (A)
        .. controls (2.3,4.5) and (2.6,4.5) .. (B) -- (C)
	.. controls (3.8,4) and (4,4) .. (D) -- (5,3) ;
      \draw (0,2) -- node[above] {\small{old $S_i$}}  (1.5,2) -- (A);
      \draw (B) .. controls +(0.2,-0.5) .. (C);
      \draw (D) -- (4.65,2) -- (5,2);
      \draw[dashed,very thin] (0,2) node [left] {\small{$t_{i-1}$}} -- (5,2);
      \draw[dashed,very thin] (0,3) node [left] {\small{$t_i$}}-- (5,3);
      \draw[<->] (1.7,2.45) -- node[below=-2pt] {\scriptsize{$>\varepsilon$}} (4.4,2.45);
    \end{tikzpicture}
  \end{center}
\label{fig:2dTest}
\end{figure}

\begin{itemize}
  \item 
    First, let $S_i (\omega) = S_{i-1}(\omega) + f_i (\omega)$, and $t_i=t_{i-1}$.
  \item 
    If after that the measure of the set 
    $\{\omega|S_i (\omega) > t_i\}$ exceeds $\varepsilon$, increase $t_i$ 
    to get rid of this excess (minimally).
  \item 
    Change $S_i$ as follows: $S_i(\omega) := \max(S_i(\omega), t_i )$
\end{itemize}

If the two last ``correction steps'' were omitted, the sequence $S_i$ 
would converge to $\sum_i f_i$. The correction steps make functions $S_i$ bigger 
(small values of $S_i$ are replaced by the threshold). 
Note that the second step is well defined, 
since $S_i$ is a basic function, and $t_i$ will be one of its finitely many values.
The following two invariant relations are easy to check:

\begin{itemize}
  \item 
    The measure of the set $\{\omega|S_i (\omega) > t_i\}$  is
    bounded by $\varepsilon$. [Indeed, the second step 
    restores this relation if it was destroyed by the
    first step, and the third step does not change 
    the set in question, since the inequality is strict.]
  \item 
    $\varepsilon t_i + \int_\Omega[S_i(\omega) - t_i]\,d\omega \leq \sum_{k=0}^i \int_\Omega f_k(\omega)d\omega$.
    [Indeed, the first step increases the integral in the left-hand side by
    $\int_\Omega f_i$, and two other steps (combined) only decrease the lefthand
    side (the horizontal sections exceeding $\varepsilon$
    are replaced by $\varepsilon$, see the illustration).]

\end{itemize}
Since the right-hand side of the last inequality
is bounded by assumption, the sequence $t_i$ is a
bounded (computable increasing) sequence, and
its limit $T = \lim t_i$  is lower semicomputable (and
therefore $\mathbf{0}'$-computable). The limit of $S_i$ is
some lower semicomputable function $S(\cdot)$.

Recall that we have to construct a $\mathbf{0}'$-effectively
open set of small measure that covers all $\omega$ where tails of $f_i$
exceed tails of all converging computable series. 
This set is defined as the set $W_\varepsilon$  
of all $\omega $ such that $S(\omega) > T $. We need to check that
this set works:

\begin{itemize}
  \item 
    $W_\varepsilon$ is $\mathbf{0}'$-effectively open (uniformly in $\varepsilon$), since $T$ is $\mathbf{0}'$-computable 
    and $S$ is lower semicomputable (even without $\mathbf{0}'$-oracle).
  \item 
    The measure of $W_\varepsilon$ does not exceed $\varepsilon$. 
    Indeed, if it does, then the measure of the set
    $\{\omega|S_i (\omega) > T \}$  would exceed $\varepsilon$ for some $i$,
    which would immediately make the threshold greater than its limit value $T$.
  \item 
    Finally, we need to show that $\omega \in W_\varepsilon$ if the tails of the
    series $\sum_i f_i(\omega)$ are not bounded by tails of any
    computable converging series. In our case
    we compare it with the convergence $t_i\rar T$, i.e., with the series $\sum (t_{i}-t_{i-1})$. 
    Indeed, our assumption guarantees that some tail 
    $f_i (\omega) + f_{i+1} (\omega) + \ldots$ exceeds the distance
    $T - t_{i-1}$, and this implies that $S(\omega) > T$  (since
    we add $f_n(\omega)$  at each step, starting from the same point $t_{i-1}$; 
    additional increases are possible, too).
\end{itemize}
\end{proof}

\subsection{Proof of Theorem~\ref{th:miller-pref}}

Now we are ready to finish the proof of Theorem~\ref{th:miller-pref} by applying Lemma~\ref{lem:coverSlowConvergence} to the sum used
in G\'acs' formula for the universal lower semicomputable test.
We already mentioned the formula for randomness deficiency:
$$
    \mathbf{d}(\omega)=\sup_n [n - \KP(\omega_1\ldots\omega_n)]+O(1).
$$
It is convenient to rewrite it in exponential form.     
Namely, let $\mathbf{m}(x)$ be the universal discrete 
semimeasure $\mathbf{m}(x) = 2^{-K(x)}$, and let $P(x)$  
be the uniform measure of the interval $x\Omega$, i.e.,
$P (x) = 2^{-|x|}$. 
Then for the universal test $\mathbf{u}(\omega)=2^{\mathbf{d}(\omega)}$ we get (up to
$O(1)$-factors in both directions)
\[
 \mathbf{u}(\omega) = \max_{x \prec \omega} \frac{\mathbf{m}(x)}{P(x)} 
\]
where the maximum is taken over prefixes $x$ of $\omega$. Gacs~\cite{GacsExactExpressions} showed not
only this formula, but also a similar formula where maximum is replaced by sum: 
\[
 \mathbf{u}(\omega) = \sum_{x \prec \omega} \frac{\mathbf{m}(x)}{P(x)} 
\]  
 (See~\cite{GacsTestsInClass} for the details.) 
 In fact, we only need to know that the right hand side of this formula has finite integral. 
 For a fixed $x$ the integral of the corresponding term is $\mathbf{m}(x)$, so the entire integral is $\sum_x \mathbf{m}(x)\le 1$.
 
To prove Theorem~\ref{th:miller-pref}, we apply Lemma~\ref{lem:coverSlowConvergence} 
to the sequence 
\[
f_i(\omega) =  \mathbf{m}(x)/P(x) = 2^{i - \KP(\omega_1\dots\omega_i)}
\]
and our assumption says that the ratio $f_i(\omega)/\mathbf{m}(i)$ tends to infinity. (Recall that $\mathbf{m}(i)=2^{-\KP(i)}$.) So the tails of the series $f_i(\omega)$ are not bounded by the tails of the series~$\mathbf{m}(i)$ and therefore not bounded by tails of any computable converging series (being maximal, $\sum\mathbf{m}(i)$ has $O(1)$-bigger tails). The theorem is proven.

\section{Prefix-free complexity: the quantitative result}
\label{sec:prefixStrong}

This section is devoted to the quantitative version of the result of the previous section.

\begin{theorem}\label{th:2randK}
  \[
    \mathbf{d}^\0(\omega) = \liminf_i [i + \KP(i) - \KP(\omega_1\ldots \omega_i)] + O(1) \,.
  \]
\end{theorem}

In the previous section we already proved the $\le$-inequality; now we need to prove the reverse
one. This follows from  Lemma~\ref{lem:ConidisPrefix} and in its proof we use
a quantitative version of Lemma~\ref{lem:coverSlowConvergence}.

\begin{lemma}\label{lem:ConidisPrefix}
  Let $f_i(\cdot)$ be a series of lower semicomputable functions on the Cantor space 
  such that $\sum_i \int f_i<\infty$. 
  Then there exist a $\0$-lower-semicomputable function
  $Q(\cdot)$ on Cantor space with finite integral such that 
  \[
    \liminf_i \left[ \frac{f_i(\omega)}{\mathbf{m}(i)} \right] \le O(Q(\omega)) \,.
  \]
\end{lemma}

The $\ge$-inequality of Theorem \ref{th:2randK} then follows from this lemma if we let (as before)
$$
f_i(\omega) = \mathbf{m}(\omega_1\dots\omega_i)/P(\omega_1\dots\omega_i) = 2^{i-\KP(\omega_1\dots\omega_i)}\,.
$$
The lemma gives us a function $Q(\cdot)$ that is a $\mathbf{0}'$-lower semicomputable test (up to a constant: the integral of $Q$ may exceed~$1$, but is finite) and 
$$
\log Q(\omega) \ge \liminf \,[\, (i - \KP(\omega_1\dots\omega_i)) + \KP(i) \,] + O(1)
$$
for every $\omega$. Since
 $\mathbf{d}^\0(\omega)$ is universal, we get the desired $\ge$-inequality.

It remains to prove Lemma \ref{lem:ConidisPrefix}. As we have done in Section~\ref{sec:prefix}, we convert
functions $f_i\colon\Omega\to\mathbb{R}$ to sets in $\Omega\times\mathbb{R}$. Then we apply a version of Lemma~\ref{lem:coverSlowConvergence} (Lemma~\ref{lem:coverSlowConvergenceGeneralized} below) to functions defined on this space.

Let us first explain what are the changes in Lemma~\ref{lem:coverSlowConvergence}.
We considered a sequence of functions $g_i(x)$ and then the set of points $x$ where the
ratios
$$
    \frac{g_i(x)+g_{i+1}(x)+\ldots}{\mathbf{m}(i)+\mathbf{m}(i+1)+\ldots}
$$ 
are not bounded (we have changed the notation and write $g_i$ instead of $f_i$ to avoid confusion, since now the lemma is applied not to $f_i$ but to other functions). The change is that now we consider a larger set of points where these ratios are not bounded by some specific constant ($1$, though any other constant would work), and cover it by a $\0$-effectively open set of finite measure. (The entire space $\Omega\times\mathbb{R}$ now has infinite measure, so this makes sense.)  Here is the exact statement:

\begin{lemma}\label{lem:coverSlowConvergenceGeneralized}
Consider a sequence of uniformly lower semicomputable non-nega\-tive functions  $g_i\colon
\Omega\times\mathbb{R}_{\ge0}\to\mathbb{R}_{\ge0}$ such that 
$\sum_i\int_{\Omega\times\mathbb{R}_{\ge0}} g_i$ is finite, where the integrals are taken with respect to
the product of standard measures on Cantor space and $\mathbb{R}_{\ge0}$. Then there exists a
$\0$-effectively open set $W \subseteq \Omega\times\mathbb{R}_{\ge0}$ of finite measure that covers
all points $z$ such that $$g_i(z) + g_{i+1}(z) + \ldots > \mathbf{m}(i) + \mathbf{m}(i+1) + \ldots$$ 
for some $i$.
\end{lemma}

In this lemma we speak about  effectively open sets and lower semicomputable functions for the space
$\Omega\times\mathbb{R}_{\ge0}$, so we need to define them formally. An \emph{effectively open set}
is a union of an enumerable family of basic open sets of the form $x\Omega\times (a,b)$ where
$x\Omega$ is an interval in the Cantor space and $(a,b)$ is an open interval with rational
endpoints; the interval $[0,b)$ can also be used instead of $(a,b)$.  A \emph{lower semicomputable
function} $g\colon\Omega\times\mathbb{R}_{\ge0}\to\mathbb{R}_{\ge0}$ can be defined as a function
such that for every rational $r$ the preimage $\{(\omega,u)|g(\omega,u)<r\}$ is effectively open
uniformly in $r$. However, for the proof it is convenient to use an equivalent definition of  lower
semicomputable functions as pointwise limit of increasing computable sequences of basic functions.
Here a \emph{basic} function is a non-negative function $b(\omega,r)$ that depends only on some
finite prefix of $\omega$ (of some length $m$) and for each of $2^m$ values of $\omega$ is a
piecewise constant function of $r$ that has finite support, and rational breakpoints and
values. Such a function is a constructive object, so we can speak about
computable sequences of basic functions in which the breakpoints and the number of breakpoints of 
each basic function are computable. Taking differences, we can also say that a lower semicomputable
function is a sum of a series whose terms are basic functions.

\begin{proof}
We use the same construction as in the proof of Lemma~\ref{lem:coverSlowConvergence} (see
figure~\ref{fig:3dKrand}), but now the
threshold $\varepsilon$ is large; we will see later how large $\varepsilon$ should be. 
Without loss, we can assume the functions $g_i$ to be  computable (rather than lower semicomputable) 
basic functions defined on $\Omega\times\mathbb{R}_{\ge0}$; 
indeed, by delaying terms, the tails only increase, making the statement only stronger.
The functions $S_i$ are now basic functions too, and $t_i$
are still rational numbers. Recall the construction: we first add $g_i$ (was $f_i$) to $S_{i-1}$,
then take minimal $t_i$ such that the set $S_i(\cdot)>t_i$ has measure at most $\varepsilon$, and
then let $S_i:=\max(S_i,t_i)$. The choice of $t_i$ now is a more difficult task, but since
$\varepsilon$ is rational, functions $S_i$ are basic, and the set $S_i(\cdot)>t_i$ is non-increasing
in $t_i$, the number $t_i$ is rational and can be computed from~$i$. 

The construction of $S_i$ and $t_i$ depend on $\varepsilon$, so we use the notation
$S_i^\varepsilon$ and $t_i^\varepsilon$ for them. The set $W^\varepsilon$ where the function $\lim
S_i^\varepsilon$ exceeds $T^\varepsilon=\lim t_i^\varepsilon$ is $\0$-effectively open uniformly in
$\varepsilon$. Note that the limit $T^\varepsilon$ is finite and the set $W^\varepsilon$ has measure
at most $\varepsilon$ (for every $\varepsilon$) for the same reasons as before; more precisely,
$T^\varepsilon=O(1/\varepsilon)$. ($T^\varepsilon \varepsilon \le \sum_i \int g_i(z)dz \le O(1)$.) 
We need only to prove that for some~$\varepsilon$ the set
$W^\varepsilon$ contains all the points $z$ such that $$g_i(z) + g_{i+1}(z) + \ldots > \mathbf{m}(i)
+ \mathbf{m}(i+1) + \ldots$$ for some $i$.

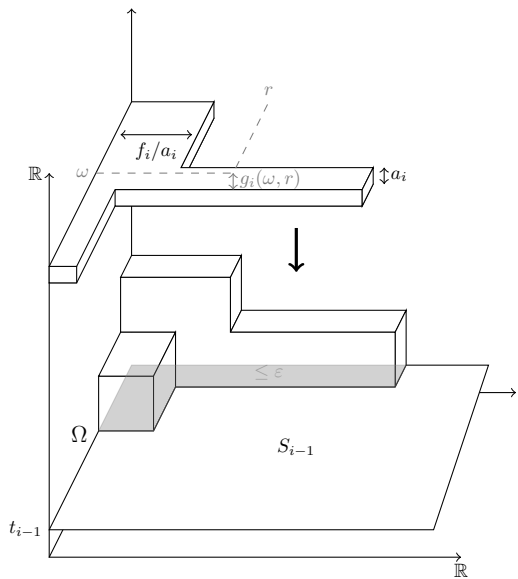
\begin{figure}[h]
  \centering
  \begin{tikzpicture}[scale=0.73, every node/.style={transform shape}]
    \draw[<->] (7,0) -- (0,0) -- node[near start,anchor=east,xshift=-3mm] {\large{$\Omega$}} (-1.5,-3) -- +(7.5,0) node[anchor=north] {$\mathbb R$};
    \draw[->] (0,0) -- (0,7);
    \draw[xshift=-1.5cm,yshift=-3cm,->] (0,0) -- (0,7) node[anchor=east] {$\mathbb R$};

    \filldraw[fill=white] (0,2.5) -- (2,2.5) -- +(-0.2,-0.4) -- ($(0,2.5) - (0.2,0.4)$) -- cycle;
    \filldraw[fill=white] (2,2.5) -- (2,1.5) -- +(-0.2,-0.4) -- ($(2,2.5) - (0.2,0.4)$) -- cycle;
    \filldraw[fill=white] (2,1.5) -- (5,1.5) -- +(-0.2,-0.4) -- ($(2,1.5) - (0.2,0.4)$) -- cycle;
    \filldraw[fill=white] (5,1.5) -- (5,.5)  -- +(-0.2,-0.4) -- ($(5,1.5) - (0.2,0.4)$) -- cycle;
    
    \filldraw[fill=white,xshift=-2mm,yshift=-4mm] (0,2.5) -- (0,1.5) -- (1,1.5) -- (1,0.5) -- (5,0.5) -- (5,1.5) -- (2,1.5) -- (2,2.5) -- cycle;
    \filldraw[fill=white,xshift=-6mm,yshift=-12mm] (0,0.5) -- +(1,0) -- +(1,1) -- +(0,1) -- cycle;
    \filldraw[fill=white,xshift=-2mm,yshift=-4mm] (0,1.5) -- (1,1.5) -- +(-.4,-.8) -- ($(0,1.5) - (.4,.8)$) -- cycle;
    \filldraw[fill=white,xshift=-2mm,yshift=-4mm] (1,1.5) -- (1,0.5) -- +(-.4,-.8) -- ($(1,1.5) - (.4,.8)$) -- cycle;

    \filldraw[fill=white] (-1.5,-2.5) -- (-0.6,-.7) -- ++(1,0) -- ++(.4,.8) -- ++(4,0) -- ++(.2,.4) -- ++(1.5,0)
    -- ++(-1,-3) -- cycle;

    \filldraw[gray,fill=gray,opacity=0.35] (0,0.5) -- node[anchor=north,darkgray,yshift=4pt] {$\le \varepsilon$} ++(5,0) -- ++(-.2,-.4) -- ++(-4,0) -- ++(-.4,-.8) -- ++(-1,0) -- cycle;

    \node at (3,-1) {$S_{i-1}$};
    \node[anchor=east] at (-1.5,-2.5)  {$t_{i-1}$};

    \draw[yshift=5cm] (0,0) -- (1.5,0) coordinate (b) -- ++(-.6,-1.2) -- ++(3.5,0) coordinate (a) -- ++(-.2,-.4) -- ++(-4.5,0) --
    ++(-.7,-1.4) -- ++(-.5,0) -- cycle;
    \draw[fill=white,yshift=5cm]  (b) -- ++(0,0.3) -- ++(-.6,-1.2) -- ++(0,-.3) -- cycle;
    \filldraw[fill=white,yshift=5cm] (-1.5,-3) rectangle +(0.5cm,0.3cm);
    \filldraw[fill=white,yshift=5cm] (-0.3,-1.6) rectangle +(4.5cm,0.3cm);
    \draw[fill=white,yshift=5.3cm] (0,0) -- (1.5,0) -- ++(-.6,-1.2) coordinate (b) -- ++(3.5,0) -- ++(-.2,-.4) -- ++(-4.5,0) --
    ++(-.7,-1.4) -- ++(-.5,0) -- cycle;
    \draw[fill=white]  (a) -- ++(0,0.3) --  ++(-.2,-.4) -- ++(0,-.3) -- cycle;
    \draw[fill=white,<->]  ($(a) + (2mm,0)$) -- node[anchor=west] {$a_i$} ++(0,0.3);

    \draw[yshift=5.3cm,<->]  (-.2,-.6) -- node[anchor=north] {$f_i/a_i$} ++(1.3,0);

    \draw[->,very thick] (3,3) -- (3,2.2);

    \draw[gray,dashed,yshift=5.3cm] (-.65,-1.3) node[anchor=east] {$\omega$} -- 
       +(2.5,0) coordinate (c) -- (2.5,0) node[anchor=south] {$r$};

    \draw[gray,<->] (c) -- node[anchor=west,yshift=1pt] {$g_i(\omega,r)$}  +(0,-0.3); 
   \end{tikzpicture}
   \label{fig:3dKrand}
   \caption{Constructing $t_i$ and $S_i$, and choice of $g_i(\omega,r)$.}
\end{figure}

This is guaranteed if 
$$m(i) + m(i+1) + \ldots \ge \Delta t_i^\varepsilon + \Delta t_{i+1}^\varepsilon + \ldots$$
where $\Delta t_i^\varepsilon$ is defined as the difference $t_i-t_{i-1}$ (in the construction for
the corresponding value of~$\varepsilon$). 
We show that $\Delta t_i^\varepsilon \le m(i)$ for large $\varepsilon$.
Since $\Delta t_i^\varepsilon$ is computable (given $i$
and $\varepsilon$) and 
   $$\sum_i \Delta t_i^\varepsilon = O(1/\varepsilon),$$
we can estimate $\Delta t_i^\varepsilon$:
   $$\Delta t_i^\varepsilon = O(\mathbf{m}(i)2^{\KP(\varepsilon)}/\varepsilon).\eqno(*)$$
Indeed, the sum
    $$\sum_{\varepsilon,i} 2^{{}-\KP(\varepsilon)}\varepsilon \Delta t_i^\varepsilon\le \sum_\varepsilon 2^{{}-\KP(\varepsilon)}\varepsilon O(1/\varepsilon)=O\bigl(\sum_\varepsilon 2^{{}-\KP(\varepsilon)}\bigr)$$
is finite, so 
    $$2^{-\KP(\varepsilon)}\varepsilon \Delta t_i^\varepsilon \le O(\mathbf{m}(i,\varepsilon))\le O(\mathbf{m}(i)).$$
Whatever the $O$-constant in $(*)$ is, we can ensure that $\Delta t_i^\varepsilon < \mathbf{m}(i)$ if
we take $\varepsilon$ large and simple enough, i.e., $\varepsilon=2^k$ for large $k$. As we have
seen, such $\varepsilon$ finishes the proof of Lemma~\ref{lem:coverSlowConvergenceGeneralized}.
\end{proof}

Using this result, we can now prove Lemma~\ref{lem:ConidisPrefix} (and therefore finish the proof of Theorem~\ref{th:2randK}).

\begin{proof}
 Let $a(i)$ be a computable sequence of rational numbers that converges slower than $\mathbf{m}(i)$ in the sense that $a(i) + a(i+1) + \ldots > \mathbf{m}(i) + \mathbf{m}(i+1) + \ldots$ for all $i$. By universality of $\mathbf{m}$, it suffices to prove the statement of the lemma where $\mathbf{m}(n)$ is replaced by $a(n)$, i.e., to construct $Q$ such that
 \[
    Q(\omega) \ge \liminf_i \left[ \frac{f_i(\omega)}{a(i)} \right]\,.
 \]
 First we construct the functions $g_i(\omega,u)$ to which
 Lemma~\ref{lem:coverSlowConvergenceGeneralized} is applied. (Remember that $\omega$ is a point in
 Cantor space, and $u$ is a non-negative real number.)  Consider the function $f_i/a(i)$ and the
 points below its graph, i.e., pairs $(\omega,u)$ such that $0\le u < f_i(\omega)/a(i)$. The area of
 this ``lower-graph'' is $\int f_i / a(i)$. Then we consider the indicator function of this set
 multiplied by $a(i)$: let $g_i(\omega,u)$ be equal to $a(i)$ if $0\le u <f_i(\omega)/a(i)$ and zero
 otherwise (see also figure~\ref{fig:3dKrand}). The integral of $g_i$ (over
 $\Omega\times\mathbb{R}$) equals $\int f_i$, so the sum of integrals is finite. The functions are
 uniformly lower semicomputable.
 
 Applying Lemma~\ref{lem:coverSlowConvergenceGeneralized}, we get a $\0$-effectively open set $W\subset\Omega\times\mathbb{R}_{\ge0}$ of finite measure that contains all pairs $(\omega,u)$ such that 
   $$g_i(\omega,u) + g_{i+1}(\omega,u) + \ldots > \mathbf{m}(i) + \mathbf{m}(i+1) + \ldots$$ 
Note that that includes all points $(\omega,u)$ such that 
   $$0\le u < \liminf_i \left[\frac{f_i(\omega)}{a(i)}\right]\,.$$
Indeed, for such $\omega$ and $u$ the point $(\omega,u)$ is under the graph of $f_i/a(i)$ for large enough $i$, so $g_i(\omega,u)=a_i$ for large enough $i$ and 
    $$g_i(\omega,u) + g_{i+1}(\omega,u) + \ldots =a(i)+a(i+1)+\ldots > \mathbf{m}(i) + \mathbf{m}(i+1) + \ldots$$ 
for large enough $i$.   
   
Now, having the $\0$-effectively open set $W$, we define the function $Q$ as a maximal function such that the area under this function is in $W$:
    $$Q(\omega)=\sup\{ v | (\omega,u)\in W \text{ for all $u$ in $[0,v)$}\}.$$
Note that this function is lower semicomputable for every effectively open $W$ with the same oracle; 
the area under its graph is included in $W$ and therefore the integral of $Q$ does not exceed the
area of $W$ and is finite. As we already noted, $Q$ is an upper bound for $\liminf$ in question.
Lemma~\ref{lem:ConidisPrefix} is proved.  
\end{proof}

\bibliographystyle{plain}
\bibliography{eigen,kolmogorov}

\end{document}